\documentclass[journal,10pt]{IEEEtran}

\ifCLASSINFOpdf

\else

\fi

\usepackage{tikz}
\usepackage{graphicx}
\usepackage{citesort}
\usepackage{amsmath, bm}
\usepackage{amsthm}
\usepackage{amssymb, epsfig}
\usepackage{epstopdf}
\usepackage{subfigure,algorithmic,float}
\usepackage{lipsum}
\usepackage{textcomp}
\usepackage{color}
\usepackage{pgfplots}
\usepackage{filecontents}
\usepackage{xcolor}
\usepackage{mathtools}
\usepackage{color}
\usepackage{stfloats}
\usepackage{framed}
\pgfplotsset{compat=newest}
\usetikzlibrary{decorations.markings}
\usetikzlibrary{patterns}
\usetikzlibrary{decorations.pathreplacing,angles,quotes}
\theoremstyle{plain}
\newtheorem{lemma}{Lemma}
\newtheorem{prop}{Proposition}
\newtheorem{cor}{Corollary}

\theoremstyle{definition}

\begin{document}

\title{\huge{Covert Wireless Communication in Presence of \\a Multi-Antenna Adversary and Delay Constraints}}

\author{{Khurram~Shahzad, \IEEEmembership{Student Member, IEEE}, Xiangyun~Zhou, \IEEEmembership{Senior Member, IEEE} and Shihao~Yan, \IEEEmembership{Member, IEEE}}

\thanks{Copyright (c) 2015 IEEE. Personal use of this material is permitted. However, permission to use this material for any other purposes must be obtained from the IEEE by sending a request to pubs-permissions@ieee.org.}

\thanks{K. Shahzad and X. Zhou are with the Research School of Electrical, Energy and Materials Engineering, Australian National University, Canberra, ACT 2601, Australia (emails: \{khurram.shahzad, xiangyun.zhou\}@anu.edu.au).}

\thanks{S. Yan is with the School of Engineering, Macquarie University, Sydney, NSW 2109, Australia (e-mail: shihao.yan@mq.edu.au).}}


\maketitle

\begin{abstract}
Covert communication hides the transmission of a message from a watchful adversary, while ensuring reliable information decoding at the receiver, providing enhanced security in wireless communications. In this letter, covert communication in the presence of a multi-antenna adversary and under delay constraints is considered. Under the assumption of quasi-static wireless fading channels, we analyze the effect of increasing the number of antennas employed at the adversary on the achievable throughput of covert communication. It is shown that in contrast to a single-antenna adversary, a slight increase in the number of adversary's antennas drastically reduces the covert throughput, even for relaxed covertness requirements.
\end{abstract}
\small

\normalsize
\IEEEpeerreviewmaketitle
\section{Introduction}
The broadcast nature of wireless transmission makes it prone to unauthorized access, raising serious concerns about its security and privacy. With an ever-increasing dependence on wireless devices not only for communication but also activities related to health, finance and sharing private information, there is a renewed interest in the security and privacy of wireless transmissions. Circumstances exist where instead of protecting the information content of the transmission, it is imperative to hide the transmission itself. In such situations, traditional security schemes employing cryptography and physical layer security \cite{sean_book} are deemed insufficient, as such schemes cause suspicion, drawing further probing from an adversary. Hiding communications in sensitive or hostile environments is of paramount importance to military and law enforcement agencies. On the other hand, detecting any malicious covert communications is also highly desired by law enforcement and cyber task forces since even the presence of such activities offers sufficient incentive for them to take action \cite{shihao2019mag}. In all such scenarios, covert communication approaches \cite{commag15bash} offer a highly viable solution, with applications not only of interest to military organizations but general public as well.

Recent research efforts have explored this nascent approach to security under different communication scenarios establishing the fundamental limits in additive white Gaussian noise (AWGN) channels \cite{bash_jsac}, under channel and noise uncertainty \cite{shahzad_vtc,biao_cc}, use of jamming and artificial noise \cite{tamara_jammer_2017,shahzad2018achieving} and under relay networks \cite{hu2018covert_relay}. Optimality of Gaussian signalling in covert communication has been analyzed in \cite{yan2019gaussian}, while \cite{yan2019tsp} offers a first study in considering a UAV as the transmitter in the context of covert communications. Multi-antenna covert communications under AWGN channels has been considered in \cite{mimo_koksal}, while a recent work in \cite{zheng2019multi} has considered their performance in random wireless networks considering both centralized and distributed antenna systems at the transmitter.

The aforementioned works consider covert communication under the assumption of an infinite number of channel uses. However, limited storage resources and requirements of quick updates in modern communication systems require a finite, sometimes small, number of channel uses, and hence the results in the infinite blocklength regime do not hold anymore. Under finite blocklength, covert communication has been considered in the literature under AWGN and fading channels \cite{yan2019delay,shu2019delay,tang2018covert}, where a single-antenna adversary was considered.

In this work, we consider covert communication under fading channels with a finite blocklength, in presence of a multi-antenna adversary. Equipping the adversary with multiple antennas makes him a stronger adversary in fading conditions, yielding the task of communicating covertly even harder. Although \cite{zheng2019multi} has considered achieving covertness for a multi-antenna transmitter in the presence of multiple adversaries, the adversaries are non-colluding, whereas we consider an adversary utilizing a centralized detection approach. Furthermore, the analysis in \cite{zheng2019multi} is presented under the assumption of an infinite blocklength as compared to our finite blocklength assumption. We analyze the achievable throughput under strict delay constraints and study the impact of adversary's multiple antennas on covertness, when the covert communication pair is equipped with single antennas. In particular, we show that in contrast to a single antenna case, the covert throughput quickly reduces to zero with a slight increase in number of antennas at the adversary, even for highly relaxed covertness requirements.

\section{System Model}
We consider a scenario where the adversary, Willie, uses $M \geq 1$ antennas for detection, whereas each of Alice and Bob is equipped with a single antenna. A communication slot is a block of time in which transmission of a message from Alice to Bob is complete. If Alice transmits in a slot, she sends a finite number, $L_M$, of symbols to Bob, with a symbol index of $l$, while Willie observes silently, looking to decide whether Alice has transmitted or not. The maximum number of symbols in a slot is denoted by $L_{\text{max}}$, and hence, $L_M \leq L_{\text{max}}$. We assume that Alice's transmitted signal samples are independent, with a distribution given by $\mathcal{CN}(0, P_a )$ while the distribution of AWGN at Bob's receiver is given by $\mathcal{CN}(0, 1)$. The additive noise samples at the different antennas at Willie are also considered to be independent, with a distribution of $\mathcal{CN}(0, 1)$. Furthermore, it is assumed that Willie's received signals and his noise are independent.

The wireless links from Alice to Bob and Alice to Willie are subject to quasi-static Rayleigh fading channels, which constitute a suitable model for covert communications scenario under NLOS communications, and is commonly adopted in the literature \cite{tamara_jammer_2017,tang2018covert}. Resultantly, corresponding to a large coherence time, the channel coefficients remain constant in a slot and change independently from one slot to the next. The vector of channel coefficients from Alice to Willie's $M$ antennas is denoted by $\mathbf{h}_{aw} \in \mathbb{C}^{M \times 1}$, and for each entry $h_{aw}$ of $\mathbf{h}_{aw}$, the mean of $|h_{aw}|^2$ is denoted by $1/\lambda$. We follow the common assumption that a secret of sufficient length is shared between Alice and Bob \cite{bash_jsac}, which is unknown to Willie. Employing random coding arguments, Alice generates codewords by independently drawing symbols from a zero-mean complex Gaussian distribution with a variance of $P_a$. Here, each codebook is known to Alice and Bob and is used only once. When Alice transmits in a slot, she selects the codeword corresponding to her message and transmits the resulting sequence.

Let $\mathbf{Y}=[\mathbf{y}_1, \mathbf{y}_2, \dots, \mathbf{y}_{L_M}] \in \mathbb{C}^{M \times L_M}$ be the matrix containing Willie's observed signals at $M$ antennas. Willie faces a binary hypothesis test regarding Alice's actions and we denote Willie's hypotheses of Alice transmitting or not by $H_1$ and $H_0$, respectively. The observation model at Willie regarding Alice's transmission state can be expressed as
\begin{equation}\label{a1}
\begin{cases}
H_0 : \mathbf{y}_l \sim \mathcal{CN}(\mathbf{0}, \mathbf{I}_M), \\
H_1: \mathbf{y}_l \sim \mathcal{CN}(\mathbf{0}, P_a \mathbf{h}_{aw}\mathbf{h}_{aw}^H + \mathbf{I}_M) ,
\end{cases}
\end{equation}
where $\mathbf{I}_M$ is an $M \times M$ identity matrix. We assume that in a given slot, $P_a$ is fixed and known to Willie. Under the assumption of an equal probability of Alice transmitting or not in a slot, achieving covertness requires $\mathbb{P}_{FA} + \mathbb{P}_{MD} \geq 1 - \epsilon$ for some arbitrarily small $\epsilon$ \cite{bash_jsac,yan2019delay}. Here, $\mathbb{P}_{FA}$ and $\mathbb{P}_{MD}$ denote the Probability of False Alarm and Probability of Missed Detection, respectively.

As per \cite{polyanskiy2010channel}, the decoding error probability at Bob is not negligible for finite blocklengths. For a given decoding error probability, $\delta$, the channel coding rate (in bits per channel use) for a given channel realization and for duration of $L_M$ symbols is given as \cite{ozcan2013throughput}

\begin{equation}
\begin{aligned}
R &\approx \log_2 \left(1 + P_a |h_{ab}|^2 \right) \\
&\qquad \qquad - \sqrt{\frac{1}{L_M} \left(1 - \frac{1}{(1 + P_a |h_{ab}|^2)^2} \right)} \frac{Q^{-1}(\delta)}{\ln 2} ,
\end{aligned}
\end{equation}
where $h_{ab}$ is the fading channel coefficient from Alice to Bob and $Q^{-1}(\cdot)$ is the inverse Gaussian Q-function. Since the transmission rate achieved here for a finite blocklength is less than the Shannon capacity, the decoding error probability, $\delta$, is considered to be less than $0.5$. We consider the amount of information (in bits), given by $ L_M R (1 - \delta)$, as our performance metric while $\mathbb{P}_{FA} + \mathbb{P}_{MD} \geq 1 - \epsilon$ is the covertness constraint. Alice being part of a wider network, periodically broadcasts pilot signals enabling other users to estimate their channels. Though this enables Bob to know his instantaneous channel from Alice, facilitating his message decoding, it also gives Willie the channel information from Alice. For ease of exposition, we assume that Alice is also aware of her channel to Bob. In comparison to \cite{zheng2019multi}, we note that in case of multiple transmit antennas at Alice, the assumption of CSI availability at all the antennas is even harder to justify since training by covert receivers may expose their existence and violate covertness requirements. We also note that while under infinite blocklength, having perfect channel state information (CSI) at the transmitter and receiver results in no decoding errors at the receiver, this is not the case under finite blocklength scenario, where decoding errors still occur even in the presence of perfect CSI at both the transmitter and the receiver.

\section{Detection at Willie}
Since Willie is aware of his channels from Alice, her transmit power and his own receiver's noise variance, he is able to design an optimal detector, which represents the worst case scenario from the covert communication design perspective. In the following, we present Willie's optimal detector and the corresponding detection error probabilities.

\begin{lemma}
The optimal detector at Willie has the decision rule given as
 \begin{equation}
\|\mathbf{h}_{aw}^H \mathbf{Y} \|^2 \underset{H_0}{\overset{H_1}{\gtrless}} \theta^* ,
\end{equation}
where
\begin{equation}
\theta^* = \frac{L_M}{2} \left(\frac{1}{P_a}+ \|\mathbf{h}_{aw} \|^2  \right) \ln \left(P_a \|\mathbf{h}_{aw} \|^2 + 1 \right)
\end{equation}
is the optimal decision threshold.
\end{lemma}
\begin{proof}
Since Willie has complete statistical knowledge of his observations, hence resorting to the Neyman Pearson criterion \cite{lehmann2006testing}, the optimal test for Willie to minimize his detection error probability is the likelihood ratio test. Under $H_0$, the pdf of the observation matrix at Willie, $\mathbf{Y}$, is given by
\begin{equation}
f(\mathbf{Y}| H_0) = \prod_{l=1}^{L} \frac{1}{\pi^M \sqrt{|K_0|}}\exp \left[- \mathbf{y}_l^H K_0^{-1}\mathbf{y}_l  \right],
\end{equation}
where $K_0 \triangleq \mathbf{I}_M$ is the covariance matrix of Willie's observations under $H_0$. We have
\begin{equation}
\begin{aligned}
f(\mathbf{Y}| H_0) &= \prod_{l=1}^{L} \frac{1}{\pi^M } \exp \left[- \mathbf{y}_l^H \mathbf{y}_l \right] = \frac{1}{\pi^{ML}} \exp \left[- \text{tr} (\mathbf{Y}\mathbf{Y}^H)  \right],
\end{aligned}
\end{equation}
where $\text{tr}(\cdot)$ denotes the trace of a matrix. By taking the logarithm of the pdf under $H_0$, we have
\begin{equation}
\mathcal{L}_0(\mathbf{Y}) = - ML \ln (\pi) - \text{tr} (\mathbf{Y}\mathbf{Y}^H) .
\end{equation}
Similarly, the pdf of $\mathbf{Y}$ under $H_1$ is written as
\begin{equation}\label{eqabc}
\begin{aligned}
f(\mathbf{Y}| H_1, \mathbf{h}, P_a) &= \prod_{l=1}^{L} \frac{1}{\pi^M \sqrt{|K_1|}}\exp \left[- \mathbf{y}_l^H K_1^{-1}\mathbf{y}_l  \right] \\
&= \frac{1}{\pi^{ML} \sqrt{|K_1|^L}} \exp \left[ -\text{tr}\left(K_1^{-1} \mathbf{Y}\mathbf{Y}^H  \right) \right]
\end{aligned}
\end{equation}
where $K_1 \triangleq P_a \mathbf{h}_{aw}\mathbf{h}_{aw}^H + \mathbf{I}_M$ is the covariance matrix of Willie's observations under $H_1$. Here, we have
\begin{equation}
|K_1| =  P_a \|\mathbf{h}_{aw}\|^2 + 1  ,
\end{equation}
and
\begin{equation}
K_1^{-1} = \mathbf{I}_M - \frac{\mathbf{h}_{aw}\mathbf{h}_{aw}^H}{\frac{1}{P_a}+ \|\mathbf{h_{aw}}\|^2} ,
\end{equation}
where $K_1^{-1}$ is obtained using Woodbury Matrix Identity for matrix inversion \cite{strang1993introduction}.

Putting in the expressions of $|K_1|$ and $K_1^{-1}$ in (\ref{eqabc}) and taking the logarithm of the pdf under $H_1$,
\begin{equation}
\begin{aligned}
\mathcal{L}_1(\mathbf{Y}) = &-ML \ln (\pi) - \frac{L}{2} \ln \left( P_a \|\mathbf{h}_{aw} \|^2 + 1 \right) \\
&- \text{tr}\left( \mathbf{Y}\mathbf{Y}^H \right) + \frac{\|\mathbf{h}_{aw}^H \mathbf{Y} \|^2}{\left(\frac{1}{P_a}+ \|\mathbf{h}_{aw} \|^2  \right) }.
\end{aligned}
\end{equation}
The log likelihood ratio (LLR) can be thus written as
\begin{equation}
\begin{aligned}
\text{LLR} &= \mathcal{L}_1(\mathbf{Y}) - \mathcal{L}_0(\mathbf{Y})  \\
&= \frac{\|\mathbf{h}_{aw}^H \mathbf{Y} \|^2}{\left(\frac{1}{P_a}+ \|\mathbf{h}_{aw} \|^2  \right) } - \frac{L}{2} \ln \left(P_a \|\mathbf{h}_{aw} \|^2 + 1 \right) .
\end{aligned}
\end{equation}
Comparing the LLR to a threshold results in the following optimal decision rule
\begin{equation}
\frac{\|\mathbf{h}_{aw}^H \mathbf{Y} \|^2}{\left(\frac{1}{P_a}+ \|\mathbf{h}_{aw} \|^2  \right) } - \frac{L}{2} \ln \left(P_a \|\mathbf{h}_{aw} \|^2 + 1 \right)  \underset{H_0}{\overset{H_1}{\gtrless}} \varphi,
\end{equation}
where $\varphi = 0$ for the LLR under the assumption of equal probability of whether Alice transmits or not. The optimal decision rule can then be obtained by a rearrangement.
\end{proof}

We note from (3) that the optimal detector at Willie is a maximum ratio combiner \cite{goldsmith2005wireless}, which assigns weightage to the observations at different antennas as per the corresponding channel gains. Thus the antennas with better channel gain from Alice have higher contribution in the decision statistic in (3). We next present the detection error probabilities at Willie under the optimal detector.

\begin{lemma}
The detection error probabilities at Willie, i.e., $\mathbb{P}_{FA}$ and $\mathbb{P}_{MD}$ are given as
\begin{equation}
\mathbb{P}_{FA} = 1 - \frac{\gamma\left(L_M, \frac{\theta^*}{\| \mathbf{h}_{aw} \|^2 }  \right)}{\Gamma\left(L_M \right)},
\end{equation}
and
\begin{equation}
\mathbb{P}_{MD} =  \frac{\gamma\left(L_M, \frac{\theta^*}{\|\mathbf{h}_{aw} \|^2 (\|\mathbf{h}_{aw} \|^2 P_a + 1)}  \right)}{\Gamma\left(L_M \right)},
\end{equation}
respectively, where $\gamma(a,b)$ is the lower incomplete Gamma function, $\Gamma(x)$ is the complete Gamma function, and $\theta^*$ is the optimal threshold of Willie's detector given in (4).
\end{lemma}
\begin{proof}
Under hypothesis $H_0$, $\mathbf{y}_l$ has a distribution given by $\mathcal{CN}(\mathbf{0}, \mathbf{I}_M )$. Conditioned on the known channel coefficients from Alice, the vector $\mathbf{h}_{aw}^H \mathbf{Y}$ has a complex Gaussian distribution given by $\mathcal{CN} (\mathbf{0}, \| \mathbf{h}_{aw} \|^2 \mathbf{I}_{L_M} )$. Thus,
\begin{equation}
\| \mathbf{h}_{aw}^H\mathbf{Y}\|^2 \sim \| \mathbf{h}_{aw} \|^2  \chi_{2L_M}^{2},
\end{equation}
where $\chi_{2L_M}^{2}$ denotes a chi-squared random variable (RV) with $2L_M$ degrees of freedom. Hence,
\begin{equation}
\begin{aligned}
\mathbb{P}_{FA} = \mathcal{P} \left[ \| \mathbf{h}_{aw}^H\mathbf{Y}\|^2 > \theta^* | H_0  \right].
\end{aligned}
\end{equation}
Under hypothesis $H_1$, $\mathbf{y}_l \sim \mathcal{CN}(\mathbf{0}, P_a \mathbf{h}_{aw}\mathbf{h}_{aw}^H + \mathbf{I}_M)$, and conditioned on the known channels, the distribution of  $\mathbf{h}_{aw}^H \mathbf{Y}$ is given by $ \mathcal{CN}\left(\mathbf{0}, \|\mathbf{h}_{aw} \|^2 (\|\mathbf{h}_{aw} \|^2 P_a + 1)\mathbf{I}_{L_M} \right)$. As a result, we have
 \begin{equation}
\| \mathbf{h}_{aw}^H\mathbf{Y}\|^2 \sim (\|\mathbf{h}_{aw} \|^2 (\|\mathbf{h}_{aw} \|^2 P_a + 1)) \chi_{2L_M}^{2},
\end{equation}
 giving
\begin{equation}
\mathbb{P}_{MD} = \mathcal{P} \left[ \| \mathbf{h}_{aw}^H\mathbf{Y}\|^2 \leq \theta^* | H_1  \right].
\end{equation}
Calculating the probabilities in (17) and (19) gives the desired result, hence completing the proof.
\end{proof}

We note here that the analysis of the optimal detector at Willie under finite blocklength in fading scenarios is quite different as compared to the infinite blocklength case. As highlighted in \cite{zheng2019multi}, for infinite observations at the adversary, uncertainties of transmitted signals and receiver noise vanish and the analysis is simplified, whereas this does not hold under the finite blocklength scenario.

\section{Alice's Approach to Achieve Covertness}
Due to the involvement of incomplete Gamma function, the detection performance at Willie does not lend itself well for further analysis. To proceed, we lower bound Willie's detection performance using Pinsker's inequality \cite{lehmann2006testing}, giving
\begin{equation}
\mathbb{P}_{FA} + \mathbb{P}_{MD} \geq 1 - \sqrt{\frac{1}{2} \mathcal{D}\left( \mathbb{P}_0^{L_M} || \mathbb{P}_1^{L_M}  \right)},
\end{equation}
where $\mathbb{P}_0^{L_M}$ and $\mathbb{P}_1^{L_M}$ denote the probability density functions (pdfs) of Willie's observation vectors under hypothesis $H_0$ and $H_1$ for $L_M$ independent channel uses, respectively, and $\mathcal{D} ( \mathbb{P}_0^{L_M} || \mathbb{P}_1^{L_M} )$ is the Kullback-Leibler (KL) Divergence from $\mathbb{P}_0^{L_M}$ to $\mathbb{P}_1^{L_M}$. As per \cite{yan2019delay}, $\mathcal{D}( \mathbb{P}_0^{L_M} || \mathbb{P}_1^{L_M} ) \leq 2 \epsilon^2$ must be ensured to guarantee $\mathbb{P}_{FA} + \mathbb{P}_{MD} \geq 1 - \epsilon$. Since Alice is unaware of her channels to Willie, she looks to minimize the expected value of $\mathcal{D}( \mathbb{P}_0^{L_M} || \mathbb{P}_1^{L_M} )$ over all possible realizations of her channels to Willie. The optimization problem at Alice is thus stated as:
\begin{equation}
\begin{aligned}
\text{P1} \quad  \underset{L_M, P_a}{\mathrm{maximize}} \quad &L_M R (1-\delta) \\
\mathrm{subject~to} \quad  &\mathbb{E}_{\|\mathbf{h}_{aw} \|^2} \left[ \mathcal{D}\left( \mathbb{P}_0^{L_M} || \mathbb{P}_1^{L_M}  \right) \right] \leq 2 \epsilon^2  , \\
&L_M \leq L_{\text{max}}.
\end{aligned}
\end{equation}
where $\mathbb{E}\left[ \cdot \right]$ denotes the statistical expectation, and is taken over all the channels from Alice to the multiple antennas at Willie. The expression for the KL-divergence for $M \geq 1$ is given as
\begin{equation}
\begin{aligned}
\mathcal{D}&\left( \mathbb{P}_0^{L_M} || \mathbb{P}_1^{L_M}  \right) = L_M  \mathcal{D}\left( \mathbb{P}_0 || \mathbb{P}_1  \right) \\
&=L_M \left( \ln \left(\|\mathbf{h}_{aw} \|^2 P_a + 1 \right) - \frac{ \|\mathbf{h}_{aw} \|^2 P_a }{\|\mathbf{h}_{aw} \|^2 P_a  + 1} \right) ,
\end{aligned}
\end{equation}
where, $\|\mathbf{h}_{aw} \|^2$ is the sum of $M$ independent exponential RVs, and for $M>1$, constitutes an Erlang RV, with pdf given by
\begin{equation}
\begin{aligned}
f_{\|\mathbf{h}_{aw} \|^2}(h; M, \lambda) &= E_M(h, \lambda)  \\
&=\frac{\lambda^M h^{M-1} e^{-\lambda h}}{\left(M-1 \right) !};  \quad h,\lambda \geq 0 .
\end{aligned}
\end{equation}

In the following, we present the optimal choice of $P_a$ and $L_M$ for Alice to achieve a desired level of covertness.
\begin{prop}
The optimal number of Alice's channel uses, $L_M^*$, in terms of Alice's transmit power, maximizing the throughput to Bob while satisfying a given covertness requirement, $\epsilon$, is given  by
\begin{equation}
L_M^* = \min \left(L_{\text{max}}, \frac{2 \epsilon^2}{g(P_a)} \right)  ,
\end{equation}
while the optimal transmit power at Alice is the solution to
\begin{equation}
\begin{aligned}
\underset{P_a}{\mathrm{maximize}} \quad & L_M^* R (1-\delta) ,
\end{aligned}
\end{equation}
where $g(P_a)$ is as given in (28).
\end{prop}
\begin{proof}
The expectation in (21) for $M>1$ is calculated as 
\begin{equation}
\begin{aligned}
\mathbb{E}_{\|\mathbf{h}_{aw} \|^2} \left[ \mathcal{D}\left( \mathbb{P}_0^{L_M} || \mathbb{P}_1^{L_M}  \right) \right]  &= L_M g(P_a) ,
\end{aligned}
\end{equation}
where the function $g(P_A)$ and meijer-G$(a,b,c,d,z)$ within $g(P_a)$ are defined as in (28) and (29), respectively \cite{matlab}. Due to the covertness constraint, P1 is now restated as
\begin{equation}\label{p1}
\begin{aligned}
\text{P1.1} \quad  \underset{L_M, P_a}{\mathrm{maximize}} \quad &  L_M R(1-\delta)\\
\mathrm{subject~to} \quad  & L_M \leq \min \left(L_{\text{max}}, \frac{2 \epsilon^2}{g(P_a)} \right) .
\end{aligned}
\end{equation}
Since $R$ is an increasing function of $P_a$ and $L_M$, the covertness requirement puts an upper limit on the number of channel uses by Alice in a block for a given $P_a$. The upper limit for $L_M$ is thus determined by the tighter bound between the covertness constraint and $L_{\text{max}}$. Accommodating the optimal $L_M$ results in the optimization problem in (25) which is of one dimension and can be solved by methods of efficient numerical search, hence concluding the proof.

 \begin{figure*}
\begin{align}
g(P_a) &= \frac{ \lambda^M }{\left(M-1 \right) !} \Bigg[\frac{1}{\lambda^M} \left(P_a \right)^2 \text{meijer-G} \left\{0, 1, [0,0,M],[],\frac{\lambda}{P_a}   \right\}  -  P_a^{(2-M)} \Gamma(M+1) e^{\frac{\lambda}{P_a}} \gamma (-M, \frac{\lambda}{P_a})   \Bigg] \\ \vspace{0.2cm}
\text{meijer-G}(a,b,c,d,z) &=  \text{meijer-G}([a_1, \dots, a_n],[a_{n+1}, \dots, a_p],[b_1, \dots, b_m],[b_{m+1}, \dots, b_q],z) \nonumber \\
&=G_{p,q}^{m,n}\left(a_1, \dots, a_p, b_1, \dots, b_q | z \right) \nonumber \\
&= \frac{1}{2 \pi i} \int \frac{\left(\prod_{j=1}^{m} \Gamma(b_j - s) \right) \left( \prod_{j=1}^{n} \Gamma(1-a_j + s)  \right)}{\left( \prod_{j=m+1}^{q} \Gamma(1-b_j + s)  \right) \left( \prod_{j=n+1}^{p} \Gamma(a_j - s)  \right)} z^s ds
\end{align}
\hrulefill
\end{figure*}
\end{proof}

We note from (24) that the optimal choice of blocklength, $L_M^*$, is a decreasing function of the transmit power used by Alice and is upper-bounded by the delay constraint determined by $L_{max}$. This fact will be more evident and explained in the numerical results section.

\begin{cor}
The optimal number of Alice's channel uses, for $M=1$, in terms of Alice's transmit power, maximizing the throughput to Bob while satisfying a given covertness requirement, is given  by
\begin{equation}
L_1^* = \min \left(L_{\text{max}}, \frac{2 \epsilon^2}{f(P_a)} \right),
\end{equation}
while the optimal transmit power at Alice can be found similarly as in Proposition $1$. Here, $f(P_a)$ is given as
\begin{equation}
\begin{aligned}
f(P_a) = - \left[ 1 + \left(1+ \frac{\lambda}{P_a} \right) e^{\frac{\lambda }{P_a}} Ei \left( - \frac{\lambda }{P_a} \right) \right],
\end{aligned}
\end{equation}
and $Ei (\cdot)$ is the Exponential Integral function. 
\end{cor}
\begin{proof}
The proof follows by simply putting $M=1$ in the results given in Proposition $1$.
\end{proof}

\section{Results And Discussions}

\begin{figure}[t!]
\centering
	\includegraphics[width=0.9\linewidth]{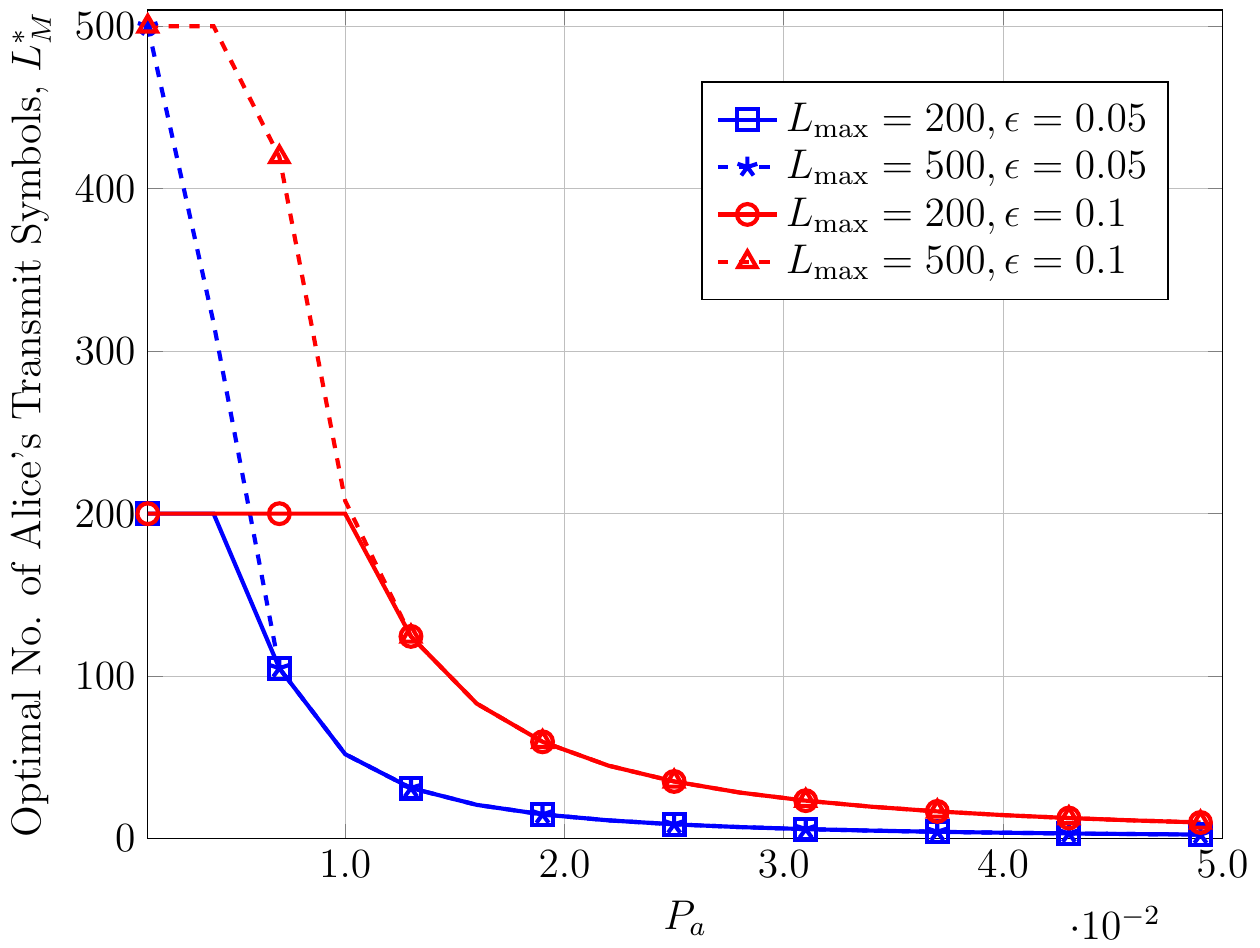}
	\caption{Number of Alice's optimal transmit symbols, $L_M^*$, versus the transmit power, $P_a$.}
	\label{fig01}
\vspace{-0.2cm}
\end{figure}

In this section, we present numerical results examining the effect of reliability and covertness on the optimal parameters for Alice and the achievable covert throughput. Unless otherwise stated, the parameters are set as follows: $\delta=0.1$, $\lambda = 1$, and $|h_{ab}|^2=1$. It should be noted here that due to the assumption of Bob and Willie's receiver noise being $\mathcal{CN}(0,1)$, the transmit power at Alice has units relative to the considered additive noise.

Fig. \ref{fig01} shows the optimal number of Alice's transmit symbols against a range of her transmit power, $P_a$, for different values of $L_{\text{max}}$ and $\epsilon$. First, the effect of limiting Alice's maximum transmit symbols is evident for lower values of $P_a$, as decreasing $P_a$ allows for Alice to use a higher blocklength under a certain covertness constraint. Furthermore, for a given value of $\epsilon$, an increase in $P_a$ requires using lower number of transmit symbols for a given covertness requirement which is in agreement with the intuition that if Alice chooses to transmit at a higher power, the transmission time should be reduced. This also represents a trade-off between achieving covertness and maximizing the throughput since the covert throughput is an increasing function of both $P_a$ and $L_M$, while achieving covertness requires a decrease in both.

In Fig. \ref{fig02}, we show the achievable throughput against a multi-antenna Willie for varying covertness requirements under the optimal choice of $L_M$ and $P_a$. We also show the achievable throughput for the case where $L_M$ is fixed while the transmit power is optimized to satisfy a required covert criteria. It should be noted here that in calculating these throughput results, we consider a lower bound on $L_M$, which is due to the use of approximated expression of $R$ as given in (2), and this bound ensures that the calculated rate is always non-negative. We note that $\epsilon=0.3$ represents a relatively relaxed covertness requirement.

As we can see, the difference between utilizing the optimal and fixed value of $L_M$ is evident in terms of the difference in throughput from Alice to Bob. Furthermore, as the number of antennas at Willie grows beyond a single antenna, there is a very sharp decrease in the achievable covert throughput in both cases and depending on the covertness requirement, it reaches zero very quickly. This shows the effectiveness of Willie in detecting any covert transmissions using more than one antenna. As we see, even for $\epsilon = 0.3$, the throughput decreases fairly quickly to zero and hence covertness can not be achieved beyond $M=16$.

\begin{figure}[t!]
\centering
	\includegraphics[width=0.9\linewidth]{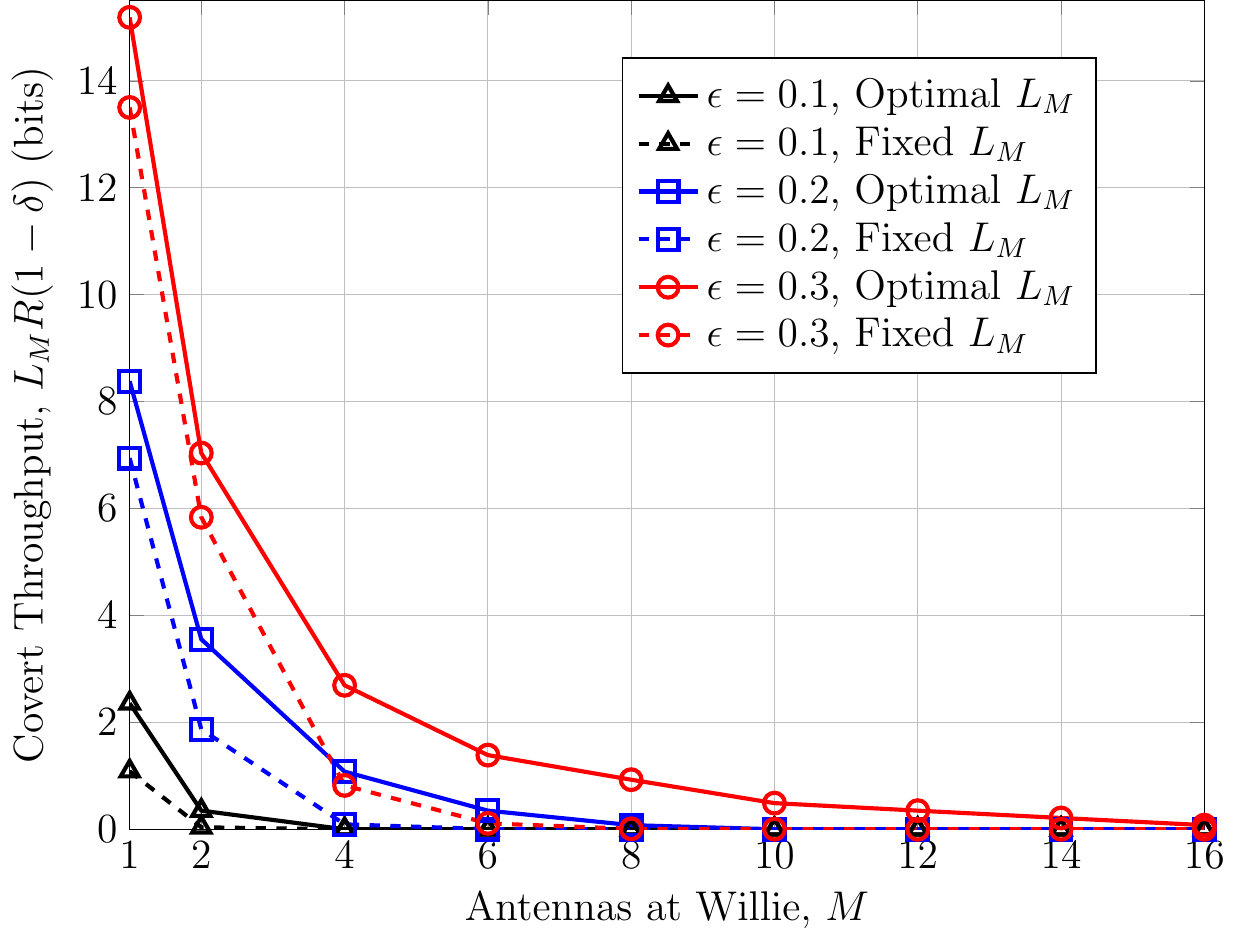}
	\caption{Achievable covert throughput from Alice to Bob, $L_M R(1-\delta)$, versus the number of antennas at Willie, $M$, for varying covertness requirements.}
	\label{fig02}
\vspace{-0.2cm}
\end{figure}

\section{Conclusion}
In this work, we have considered the performance of covert communication in the presence of a multi-antenna adversary while under strict delay constraints. We have analyzed the effect of covertness requirement and the number of antennas at Willie on Alice's achievable covert throughput. It has been shown that the improved detection capability of Willie drastically reduces the covert throughput. This letter presented an initial work on covert communication under delay constraints and in the presence of a multi-antenna adversary. Future work will focus on devising improved covertness schemes for achieving non-zero throughput despite the higher number of antennas at Willie, and under more complex fading models encompassing both LOS and NLOS scenarios.


\begin{thebibliography}{10}
\providecommand{\url}[1]{#1}
\csname url@samestyle\endcsname
\providecommand{\newblock}{\relax}
\providecommand{\bibinfo}[2]{#2}
\providecommand{\BIBentrySTDinterwordspacing}{\spaceskip=0pt\relax}
\providecommand{\BIBentryALTinterwordstretchfactor}{4}
\providecommand{\BIBentryALTinterwordspacing}{\spaceskip=\fontdimen2\font plus
\BIBentryALTinterwordstretchfactor\fontdimen3\font minus
  \fontdimen4\font\relax}
\providecommand{\BIBforeignlanguage}[2]{{%
\expandafter\ifx\csname l@#1\endcsname\relax
\typeout{** WARNING: IEEEtran.bst: No hyphenation pattern has been}%
\typeout{** loaded for the language `#1'. Using the pattern for}%
\typeout{** the default language instead.}%
\else
\language=\csname l@#1\endcsname
\fi
#2}}
\providecommand{\BIBdecl}{\relax}
\BIBdecl

\bibitem{sean_book}
X.~Zhou, L.~Song, and Y.~Zhang, \emph{Physical Layer Security in Wireless
  Communications}.\hskip 1em plus 0.5em minus 0.4em\relax CRC Press, 2013.

\bibitem{shihao2019mag}
S.~Yan, X.~Zhou, J.~Hu, and S.~Hanly, ``Low probability of detection
  communication: Opportunities and challenges,'' \emph{IEEE Wireless Commun.,
  To appear, arXiv:1803.07812}, 2019.

\bibitem{commag15bash}
B.~A. Bash, D.~Goeckel, D.~Towsley, and S.~Guha, ``Hiding information in noise:
  Fundamental limits of covert wireless communication,'' \emph{IEEE Commun.
  Mag.}, vol.~53, no.~12, pp. 26--31, Dec. 2015.

\bibitem{bash_jsac}
B.~A. Bash, D.~Goeckel, and D.~Towsley, ``Limits of reliable communication with
  low probability of detection on {A}{W}{G}{N} channels,'' \emph{IEEE J. Sel.
  Areas Commun.}, vol.~31, no.~9, pp. 1921--1930, Sep. 2013.

\bibitem{shahzad_vtc}
K.~Shahzad, X.~Zhou, and S.~Yan, ``{C}overt communication in fading channels
  under channel uncertainty,'' in \emph{IEEE VTC}, Jun. 2017, pp. 1--5.

\bibitem{biao_cc}
B.~He, S.~Yan, X.~Zhou, and V.~K.~N. Lau, ``On covert communication with noise
  uncertainty,'' \emph{IEEE Commun. Lett.}, vol.~21, no.~4, pp. 941--944, Apr.
  2017.

\bibitem{tamara_jammer_2017}
T.~V. Sobers, B.~A. Bash, S.~Guha, D.~Towsley, and D.~Goeckel, ``Covert
  communication in the presence of an uninformed jammer,'' \emph{IEEE Trans.
  Wireless Commun.}, vol.~16, no.~9, pp. 6193--6206, Sep. 2017.

\bibitem{shahzad2018achieving}
K.~Shahzad, X.~Zhou, S.~Yan, J.~Hu, F.~Shu, and J.~Li, ``Achieving covert
  wireless communications using a full-duplex receiver,'' \emph{IEEE Trans.
  Wireless Commun.}, vol.~17, no.~12, pp. 8517--8530, Dec. 2018.

\bibitem{hu2018covert_relay}
J.~Hu, S.~Yan, X.~Zhou, F.~Shu, J.~Li, and J.~Wang, ``Covert communication
  achieved by a greedy relay in wireless networks,'' \emph{IEEE Trans. Wireless
  Commun.}, vol.~17, no.~7, pp. 4766--4779, Jul. 2018.

\bibitem{yan2019gaussian}
S.~Yan, Y.~Cong, S.~V. Hanly, and X.~Zhou, ``Gaussian signalling for covert
  communications,'' \emph{IEEE Trans. Wireless Commun.}, vol.~18, no.~7, pp.
  3542--3553, Jul. 2019.

\bibitem{yan2019tsp}
X.~Zhou, S.~Yan, J.~Hu, J.~Sun, J.~Li, and F.~Shu, ``Joint optimization of a
  {U}{A}{V}'s trajectory and transmit power for covert communications,''
  \emph{IEEE Trans. Signal Process.}, vol.~67, no.~16, pp. 4276--4290, Aug.
  2019.

\bibitem{mimo_koksal}
A.~Abdelaziz and C.~E. Koksal, ``Fundamental limits of covert communication
  over {M}{I}{M}{O} {A}{W}{G}{N} channel,'' in \emph{IEEE CNS}, Oct. 2017, pp.
  1--9.

\bibitem{zheng2019multi}
T.~X. Zheng, H.~M. Wang, D.~W.~K. Ng, and J.~Yuan, ``Multi-antenna covert
  communications in random wireless networks,'' \emph{IEEE Trans. Wireless
  Commun.}, vol.~18, no.~3, pp. 1974--1987, Mar. 2019.

\bibitem{yan2019delay}
S.~Yan, B.~He, X.~Zhou, Y.~Cong, and A.~L. Swindlehurst, ``Delay-intolerant
  covert communications with either fixed or random transmit power,''
  \emph{IEEE Trans. Inf. Forensics Security}, vol.~14, no.~1, pp. 129--140,
  Jan. 2019.

\bibitem{shu2019delay}
F.~Shu, T.~Xu, J.~Hu, and S.~Yan, ``Delay-constrained covert communications
  with a full-duplex receiver,'' \emph{IEEE Wireless Commun. Lett.}, vol.~8,
  no.~3, pp. 813--816, Jan. 2019.

\bibitem{tang2018covert}
H.~Tang, J.~Wang, and Y.~R. Zheng, ``Covert communications with extremely low
  power under finite block length over slow fading,'' in \emph{IEEE Info{C}om
  Workshops}, Apr. 2018, pp. 657--661.

\bibitem{polyanskiy2010channel}
Y.~Polyanskiy, H.~V. Poor, and S.~Verd{\'u}, ``Channel coding rate in the
  finite blocklength regime,'' \emph{IEEE Trans. Inf. Theory}, vol.~56, no.~5,
  pp. 2307--2359, May 2010.

\bibitem{ozcan2013throughput}
G.~Ozcan and M.~C. Gursoy, ``Throughput of cognitive radio systems with finite
  blocklength codes,'' \emph{IEEE J. Sel. Areas Commun.}, vol.~31, no.~11, pp.
  2541--2554, Nov. 2013.

\bibitem{lehmann2006testing}
E.~L. Lehmann and J.~P. Romano, \emph{Testing statistical hypotheses}.\hskip
  1em plus 0.5em minus 0.4em\relax Springer Science \& Business Media, 2006.

\bibitem{strang1993introduction}
G.~Strang, \emph{Introduction to Linear Algebra}.\hskip 1em plus 0.5em minus
  0.4em\relax Wellesley-Cambridge Press Wellesley, MA, 1993.

\bibitem{goldsmith2005wireless}
A.~Goldsmith, \emph{Wireless communications}.\hskip 1em plus 0.5em minus
  0.4em\relax Cambridge university press, 2005.

\bibitem{matlab}
{{M}{A}{T}{L}{A}{B} {R}2018a, The MathWorks, Inc., Massachusetts, United
  States.}

\end{thebibliography}
\end{document}